\documentclass{scrartcl}

\PassOptionsToPackage{nameinlink}{cleveref}

\usepackage[todo,firstnames]{group}

\usepackage[margin=2.5cm]{geometry}

\usepackage{mathtools}
\usepackage{theoremref}
\usepackage{pgfplots}
\usepackage{cancel}
\usepackage{xcolor}
\usepackage{graphicx} 

\hypersetup{
	pdfencoding=auto, 
	psdextra,
	colorlinks=true,
	citecolor=green!50!black,
	linkcolor=red!60!black,
	urlcolor=blue!90!black
}

\usepackage{subcaption}
\usepackage{makecell} 

\title{Jointly Satisfying Pareto Optimality and Justified Representation is NP-Hard in Approval-Based Multiwinner Voting}

\usepackage{authblk}

\author[]{Chris Dong}
\affil[]{Hasso Plattner Institute, University of Potsdam}
\date{\vspace{-1cm}}

\begin{document}

\maketitle

\begin{abstract}	
An open problem in approval-based multiwinner voting concerns whether we can efficiently compute committees that satisfy both justified representation and Pareto optimality. We answer this question negatively by proving that, on the domain of all profiles, outputting a committee satisfying both axioms is NP-hard.
An initial proof was found by ChatGPT Astra. This was then verified and rewritten by the author.
\end{abstract}


\section{Introduction}
In approval-based multiwinner voting, the goal is to select a committee, i.e., a fixed-size subset of the candidate set, given voters' approvals over the candidates. Here, intuitively, an approval indicates that the voter likes the candidate. Formally, this is modelled via voting rules which take as input the voters' approvals and output a committee.
Two canonical goals in this setting regard fairness and efficiency of the selected committee: justified representation (JR) requires that each large enough group of voters agreeing on at least one candidate is represented by at least one committee member. On the other hand, Pareto optimality (PO) requires that the committee leaves no room for clear improvements: there should not exist another committee that grants some voters strictly more approved candidates without decreasing the approval count of other voters.
Obtaining a PO committee is possible in polynomial time \citep{caragiannis2010approximation}, and obtaining a JR committee can be done in polynomial time, too \citep{ABC+16a}. Further, PO and JR are compatible: voting rules such as Proportional Approval Voting (PAV) are known to always return committees satisfying both axioms \citep{ABC+16a}. However, no polynomial-time computable rules are known at the time of this work that achieve both PO and JR. Therefore,
a natural open question asks for the existence of polynomial-time computable voting rules that always return committees satisfying PO and JR \citep{peters2025facets,LaSk23a,schunke2026pareto}.

\paragraph{Contribution.} We show that it is NP-hard to compute any committee satisfying JR and PO. More precisely, we prove that any voting rule that only outputs committees satisfying JR and PO is NP-hard to compute.

\paragraph{Related Work.} There are several strands of broadly related work. Many axioms modelling proportional representation have been proposed and thoroughly studied \citep[see, e.g., the work of][]{ABC+16a,SFF+17a,BrPe23a}. Consequently, many voting rules were proposed and analyzed regarding their computational properties and proportionality guarantees \citep[see, e.g., the work of][]{aziz2018complexity,PeSk20a,SFL16a,BFJL16a,BrPe23a}.
More closely related, among committees satisfying JR, \citet{EFI+24a} study the complexity of maximizing welfare (i.e., the sum of voters' approval counts) and provide hardness for approximation of this objective. Further, \citet{aziz2020computing} show that, given a committee and an instance, deciding whether the committee satisfies PO is coNP-complete.

On restricted domains, positive results have been achieved: \citet{schunke2026pareto} proves that a strengthening of JR called EJR+ can be satisfied with PO in polynomial time on the candidate interval domain and the voter interval domain. Similarly, \citet{MaSo26a} prove that so-called Thiele rules, including PAV, can be computed on the voter interval domain in polynomial time. PAV satisfies both EJR+ and PO. \citet{ALSV26a} provide polynomial-time computation of PAV on the more general domain of linearly consistent preferences.

Finally, \citet{BBFG26a} study a strengthening of PO, namely fractional Pareto optimality (fPO), under which a committee can be Pareto-dominated by a distribution over committees. They prove that no Thiele rule satisfies both JR and fPO, but leave open whether arbitrary rules can satisfy both axioms jointly. \citet{JLT26a} answer this question, proving that JR is incompatible even with a weaker fractional variant of fPO.

None of the above results directly imply an answer to the question of whether JR and PO are jointly satisfiable in polynomial time on the full domain allowing all approval profiles. On a technical level, the use of fill-up candidates in the proof of \Cref{theorem} is inspired by the impossibility result of \citet{JLT26a}.

\section{Preliminaries}
Let a set of $n$ \emph{voters} $V$ and a set of $m$ \textit{candidates} $C$ be given. Each $v\in V$ submits an \textit{approval set} $A_v \subseteq C$ indicating the candidates she likes, and the collection of all approval sets $A=(A_v)_{v\in V}$ is called an \textit{approval profile}. We write $[i]\coloneqq \{1,\dots, i\}$. An \textit{approval based committee (ABC) voting instance} $(A,k)$ is given via an approval profile $A$ and a \emph{target size} $k\in [m]$. 
A \emph{committee} $W\subseteq C$ is a collection of candidates of target size $k$. 
An \textit{ABC voting rule} $f$ takes as input instances $(A,k)$ and outputs a committee $f(A,k)$.

How can we determine whether the committee $W= f(A,k)$ is a ``good'' or a ``fair'' committee? One way to exclude intuitively weak committees is to use the following notion of optimality.
\begin{definition}
Given an instance $(A,k)$, a committee $W$ \emph{Pareto dominates} another committee $W'$, if $\lvert A_v\cap W\rvert \ge \lvert A_v\cap W'\rvert$ for all $v\in V$, and further $\lvert A_v\cap W\rvert > \lvert A_v\cap W'\rvert$ for some voter $v\in V$.
A committee $W$ is \emph{Pareto optimal}, if no committee $W'$ Pareto dominates it.
\end{definition}
However, Pareto optimality alone does not ensure that a committee is fair. We therefore consider a standard notion of proportional representation.
For this purpose, we say that a group of voters $S\subseteq V$ is \emph{$\ell$-cohesive} if $\lvert S \rvert \ge \ell \frac nk$ and $\lvert\bigcap_{v\in S}A_v \rvert \ge \ell$. 
Intuitively, the following notion requires that large enough groups of voters who agree on sufficiently many candidates should be represented by the committee accordingly.

\begin{definition}[\citealp{SFF+17a}]
Given an instance $(A,k)$, a committee $W$ satisfies \emph{proportional justified representation (PJR)}, if for all $\ell \in [k]$ and all $S\subseteq V$ that are $\ell$-cohesive, it holds that $\lvert \bigcup_{v\in S} A_v \cap W \rvert \ge \ell$.
\end{definition}

With efficiency and fairness being two natural desiderata, the following question has been posed:
\begin{quote}
   It is still unclear which combinations of axiomatic properties of ABC rules can be achieved in
polynomial time. [\dots] For example, [\dots] is there a polynomial-time computable ABC rule that is proportional (e.g., that satisfies PJR) and satisfies Pareto
optimality? \citep[][Q5]{LaSk23a}   
\end{quote}

We answer this question negatively.\footnote{The answer would be trivial if the axioms were not jointly satisfiable. Note, however, that PO and PJR are known to be compatible, i.e., there always exists a committee satisfying both notions. For instance, \citet{ABC+16a} have proven that a rule called Proportional Approval Voting (PAV) only outputs committees satisfying both axioms.} 
In fact, we prove a stronger impossibility result by considering a weaker form of proportionality. Intuitively, the following notion only demands that large enough and cohesive groups are represented by at least one candidate.

\begin{definition}[\citealp{ABC+16a}]
Given an instance $(A,k)$, a committee $W$ satisfies \emph{justified representation (JR)}, if for all $c\in C \setminus W$ and all $S\subseteq V$ with $c\in A_v$ for all $v\in S$ of size at least $\lvert S\rvert \ge \frac nk$, it holds that there exists $v\in S$ with $A_v\cap W \neq \emptyset$.
\end{definition}

Indeed, PJR can easily be seen to imply JR, and \citet{peters2025facets} observes that ``there are no known polynomial-time computable rules satisfying even JR and Pareto optimality together.''




\section{Main Result}
We are ready to state the main result.

\begin{theorem}\label{theorem}
    If $f$ is an ABC voting rule such that for all instances $(A,k)$, the committee $f(A,k)$ satisfies Pareto optimality and JR, then $f$ is NP-hard to compute. 
\end{theorem}
In fact, \Cref{theorem} still holds if we substitute JR  with any strengthening of it that is compatible with Pareto optimality. Therefore, \Cref{theorem} entails that there exists no tractable rule satisfying PJR and PO unless P=NP.

The remainder of this section is dedicated to the proof of the result. First, we require the following combinatorial lemma:
\begin{lemma}\label{lem:3satisfaction}
    Let $B$ be a set of elements, $b\in \mathbb N$, such that each $v\in B$ has a natural number $\delta_v \in \{0,\dots, b\}$ with $\sum_{v\in B} \delta_v = 3b$. Then, there exists $B^1, \dots, B^b\subseteq B$, each of size $3$, such that\footnote{possibly the same set multiple times} for each $v\in B$, $v\in B^j$ holds precisely for $\delta_v$ many indices $j\le b$.
\end{lemma}
\begin{proof}
    For $b=0$, the statement is trivially true.
    Let $b>0$. There exist at least $3$ elements $v\in B$ with $\delta_v>0$, as otherwise by $\delta_{v'}\le b$ we would have that $\sum_{v'\in B} \delta_{v'}\le 2b$. Therefore, choose three elements $v,w,x\in B$ with largest $\delta$-values. Set $B^1=\{v,w,x\}$.
    Then, we set $\delta'_{v'}= \delta_{v'}$ for all $v'\neq v,w,x$, and $\delta'_{y}= \delta_{y}-1$ for all $y\in B^1$. We verify that we can apply induction: clearly, $\sum_{v'\in B} \delta'_{v'} = 3b -3 = 3(b-1)$. Further, since $\sum_{v\in B} \delta_v = 3b$, there can be at most $3$ elements with $\delta_{v'}= b$. Since we subtract $1$ from the three largest demand values, we have $\delta'_{v'}\le b-1$ for all $v'\in B$. Induction now gives us a list of sets $B^2, \dots, B^{b}$ meeting the demands of the vector $\delta'$. This means that $B^1, B^2, \dots, B^{b}$ constitutes a solution of the original problem.
\end{proof}

We are now ready to present the proof of \Cref{theorem}.
\begin{proof}
    \textbf{The reduction problem.} We consider the problem \textsc{Exact Cover by 3-Sets (X3C)}, which is known to be NP-complete \citep{garey1990guide}. Formally, in this problem we are given some $t\in \mathbb N$, a universe $U$ of $3t$ elements, and a family $\mathcal U \subseteq \binom{U}{3}$ of size-three-subsets of $U$. Given such a X3C-instance, the question is whether there exists a selection $\mathcal T \subseteq \mathcal U$ of size $t$ partitioning $U$.

    \textbf{Building an ABC voting instance.} Given an \textsc{X3C} instance, we build the following ABC voting instance $(A,k)$, where we rebuild the X3C problem six times and add an additional voter from which we read off the answer. For each element $u\in U$, we instantiate \emph{$U$-voters} $v^u_{x}$, $x\in [6]$. We further add eleven \emph{padding voters} $v^p_1,\dots,v^p_{11}$, and one \emph{answer voter} $v^*$. Call $B\coloneqq V\setminus \{v^*\}$ the set of \emph{X3C voters}, which consists of $18t + 11$ voters (whereas $n= 18t +12$).
    For each admissible three element set $T\in \mathcal U$, we instantiate \emph{$\mathcal U$-candidates} $c^T_{x}$, $x\in [6]$. For each $x\le 6$, $c^T_{x}$ is approved by the three $U$-voters $v^u_{x}$ with $u\in T$, and by $v^*$. Therefore, for each $x$, each exact set cover $\mathcal T\subseteq \mathcal U$ corresponds to the candidate set $C^{\mathcal T}_x= \{c^{T}_{x}\in C \mid T\in \mathcal T\}$ such that each $v^u_{x}$ approves precisely one candidate from it.
    Next, for each $D\subseteq B$ with $\lvert D\rvert = 3$, create \emph{fill-up candidates} $d^D_{y}$, $y\in [4]$ approved by $D$, i.e., for all $v\in V$, $d^D_{y}\in A_v$ iff $v\in D$. As target size, we set $k= 6t+4$. Therefore, $\frac nk = 3$.\footnote{Note that the created ABC voting instance is indeed polynomial in the input size, as it contains $6 \lvert \mathcal U \rvert + 4 \binom{18t+11}{3}$ many candidates.}

    \textbf{The equivalence.} We claim the following: let $W^*= f(A,k)$ be any committee satisfying JR and PO. Then, the original X3C instance was a YES instance if and only if $\lvert W^*\cap A_{v^*}\rvert \ge 6t-2$. I.e., by querying the number of approvals the answer voter obtains from $W^*$, we can decide the X3C problem.

    To prove this claim, fix any committee $W^*$ satisfying PO and JR.

    ``$\Longrightarrow$'' First, for the more involved implication, let the X3C instance be a YES instance. There exists an exact 3-cover $\mathcal T\subseteq \mathcal U$.
    Our goal is to construct a committee $W$ such that $\lvert A_v\cap W^*\rvert = \lvert A_v\cap W\rvert$ for all $v\in B$, and $\lvert A_{v^*}\cap W\rvert \ge 6t-2$. Then, by Pareto-optimality of $W^*$, it holds that $\lvert A_{v^*}\cap W^*\rvert \ge 6t-2$. To construct $W$, we will use the sets $C^{\mathcal T}_x$ for $x\in [6]$. Informally, the idea is to select all candidates corresponding to $\mathcal T$, then discard the candidates that are approved by voters $v\in B$ with $A_v\cap W^* = \emptyset$, and then to fill the remaining seats to restore the X3C voters' approval counts. 
    
    For this purpose, first consider $B_0=\{v\in B\mid A_v \cap W^* = \emptyset\}$. We derive two useful bounds. Since $W^*$ satisfies JR and $\frac nk = 3$, we have 
    \begin{equation}\label{eq:1}
        \lvert B_0 \rvert \le 2\text{.}    
    \end{equation}
    Further, note that each candidate of either type,  $c^T_{x}$ and $d^D_{x}$, is approved by three X3C voters. Therefore, we have that $\sum_{v\in B} \lvert A_v\cap W^*\rvert = 3k$.
    By choice of our parameters, $n= 3k$ and there are precisely $3k-1$ X3C voters. Of these, $\lvert B_0 \rvert$ approve of zero candidates in $W^*$, therefore $3k - 1 - \lvert B_0 \rvert$ X3C voters approve of at least $1$ candidate in $W^*$. We now count how many approvals X3C voters may obtain beyond their first approval, via $\sum_{v\in B} \max(\lvert A_v\cap W^*\rvert -1, 0) = 3k - (3k-1-\lvert B_0 \rvert) = 1+ \lvert B_0 \rvert \le 3$.
    In particular, 
    \begin{equation}\label{eq:2}
        \lvert A_v\cap W^* \rvert \le 4 \quad \forall v\in B\text{.}
    \end{equation}

    We are ready to make our construction of $W$ rigorous: first, set $P_0 = \{c^T_{x}\mid T\in \mathcal T, x\in [6]\}$. We reduce the candidate set $P_1= P_0 \setminus \bigcup_{v\in B_0} A_v$. We now aim to apply \Cref{lem:3satisfaction} to fill up $P_1$ as required. For this purpose, we set $\delta_v \coloneqq \lvert A_v \cap W^*\rvert - \lvert A_v \cap P_1\rvert $ for all $v\in B$. We claim that $\delta_v \ge 0$ for all $v\in B$: each padding voter only approves of fill-up candidates, whereas $P_1\subseteq P_0$ consists only of $\mathcal U$-candidates. Further, each $v^u_x$ approves of precisely one candidate from $P_0$, namely $c^T_{x}$ for the unique $T\in\mathcal T$ containing $u$. Therefore, all $U$-voters $v= v^u_x\notin B_0$ have $\lvert A_v \cap W^*\rvert \ge 1 = \lvert A_v \cap P_0\rvert \ge \lvert A_v \cap P_1\rvert$, whereas trivially for all $v\in B_0$ we have $\lvert A_v \cap P_1\rvert = 0 \le \lvert A_v \cap W^*\rvert$. 
    Further $\delta_v \le 4$ by Equation \ref{eq:2}. To fill up $P_1$, we set $b\coloneqq k-\lvert P_1 \rvert$. And since $\lvert P_1 \rvert \leq \lvert P_0\rvert = 6t$ and $k= 6t+4$, this implies $b \ge 4\ge \delta_v \ge 0$. Recall that each $c\in C$ is approved by three X3C-voters. Therefore, $\sum_{v\in B} \delta_v = 3k - 3\lvert P_1 \rvert = 3 (k-\lvert P_1 \rvert ) = 3b$.  \Cref{lem:3satisfaction} yields a sequence of triples $B^1,\dots, B^b\subseteq B$ (each possibly occurring multiple times) such that $\lvert \{i\le b\mid v\in B^i\}\rvert = \delta_v$ for each $v\in B$. Start with $P_2= \emptyset$. For each set $D\in \{B^1,\dots, B^b \}$, we set $r(D) = \lvert \{i\le b \mid D= B^i\}\rvert \ge 1$. We add $d^D_{1},\dots, d^D_{r(D)}$ to $P_2$. Since there exist only four fill-up candidates for any given triplet $D\subseteq B$, we show that $r(D)\le 4$: if a set appears $r(D)$ times, then for all $v$ in this set we have $\delta_{v}\ge r(D)$. However, we know that $\delta_{v}\le 4$, proving the claim.
    This yields the desired set $P_2$ of size $b$, such that $\lvert A_v\cap P_2 \rvert = \lvert \{i\le b \mid v \in  B^i\} \rvert = \delta_v$. 
    Finally, we can set $W= P_1 \cup P_2$, and by definition obtain precisely that $\lvert A_v\cap W \rvert = \lvert A_v\cap P_1 \rvert + \lvert A_v\cap P_2 \rvert = \lvert A_v\cap W^* \rvert -\delta_v + \delta_v = \lvert A_v\cap W^* \rvert$ for all $v\in B$. Further, we claim that we remove at most $2$ candidates from $P_0$ to obtain $P_1$. For this, recall that each $v\in B_0$ that is a padding voter does not approve of any $\mathcal U$-candidates and therefore does not affect $P_0$; and that each $v\in B_0$ that is a $U$-voter approves of precisely $1$ candidate from $P_0$. Together with $\lvert B_0 \rvert\le 2$ from Equation \ref{eq:1}, this yields the claim.
    Applying this and the fact that $v^*$ only approves of $\mathcal U$-candidates, we obtain $\lvert A_{v^*} \cap W \rvert = \lvert A_{v^*} \cap P_1 \rvert  = \lvert P_1 \rvert  \ge \lvert P_0 \rvert -2 \ge 6t-2$.

    ``$\Longleftarrow$'' For the remaining implication, let the answer voter satisfy $\lvert A_{v^*} \cap W^* \rvert \ge 6t-2$. Our goal is to construct from $W^*$ a set $Q$ of $\mathcal U$-candidates such that $\lvert Q \rvert >6(t-1)$ and for all $x\in [6]$ and all distinct $c^T_{x}, c^{T'}_{x}\in Q$ we have $T\cap T' = \emptyset$. If we find such a set $Q$, then by the pigeonhole principle, there exists $x\le 6$ such that $\lvert \{c^T_{x}\in Q\mid T\in \mathcal U\} \rvert \ge t$, and we can take the corresponding $t$ sets $T$ which form a solution to \textsc{X3C}. 

    Let $Q_0\subseteq W^*$ be the set of all $\mathcal U$-candidates contained in $W^*$. We will iteratively remove candidates from $Q_0$ to remove candidates with \emph{overlap}, i.e., $c^T_x \neq c^{T'}_x$ with $T\cap T' \neq \emptyset$. Set $i=0$. Denote $r_i(v)\coloneqq \lvert Q_i \cap A_v \rvert$. Then, if two candidates with index $x$ in $Q_i$ overlap  due to $u\in T\cap T'$, then $r_i(v^{u}_x) \ge 2$. Generally, for $v\in B$, we therefore consider the excess approval count $\max(r_i(v)-1,0)$, which is positive if and only if two candidates in $Q_i$ overlap due to $v$. Since we know that $\sum_{v\in B} \max(\lvert A_v\cap W^*\rvert -1,0) \le 3$, and $Q_i\subseteq W^*$, we have that $\sum_{v\in B} \max(r_i(v) -1,0)  \le 3$. If two candidates $c^T_{x}, c^{T'}_{x}$ overlap, delete one of them from $Q_i$ and, e.g., set $Q_{i+1}= Q_i\setminus \{c^{T'}_x\}$. This means that $\sum_{v\in B} \max(r_{i+1}(v) -1,0) \le \sum_{v\in B} \max(r_{i}(v) -1,0) -1$ due to the voters $v^u_x$ approving $c^{T'}_x$ and satisfying $r_{i}(v^u_x) -1>0$ for all $u\in T\cap T'\subseteq U$. Increment $i$ by one.
    After at most three deletions, we arrive at some $Q$ and have that the sum is zero, therefore, w.r.t. $Q$, $\max(\lvert Q\cap A_v\rvert-1, 0)= 0$ for all $v\in B$. This implies that for all $x\le 6$ and distinct $c^T_{x}, c^{T'}_{x}\in Q$, we have that $T\cap T' = \emptyset$. Further, since we deleted at most three candidates, $\lvert Q \rvert \ge \lvert Q_0\rvert -3$. Since the answer voter approves precisely of all $\mathcal U$-candidates, we obtain $Q_0= A_{v^*}\cap W^*$. By assumption, therefore $\lvert Q \rvert \ge \lvert Q_0\rvert -3 \ge (6t-2)-3 = 6(t-1)+1 $.  
\end{proof}

\begin{remark}
    The idea to pad the committee using candidates that correspond to triplets of $X3C$ voters is inspired by a construction of \citet{JLT26a}, who use a similar gadget.
\end{remark}

\section*{Acknowledgments}
I thank the chair of Algorithmic Decision Making and Society at HPI, led by Niclas Boehmer, for funding my ChatGPT Pro license. 
The author wrote the manuscript and takes full responsibility for the correctness of this paper.

\bibliography{group,EJRPO/myadditions}
\end{document}